\documentclass[copyright]{eptcs}

\usepackage{breakurl}
\usepackage{amsmath}
\usepackage{amsthm}

\newtheorem{theorem}{Theorem}
\newtheorem{definition}[theorem]{Definition}
\newtheorem{lemma}[theorem]{Lemma}

\usepackage{amsfonts}
\usepackage{enumerate}
\usepackage{stmaryrd}
\usepackage{graphicx}
\usepackage{color}

\newcommand{\lra}{\longrightarrow}
\newcommand{\nn}{\ensuremath{\mathbb{N}}}

\newcommand{\rr}{\ensuremath{\mathbb{R}}}

\newcommand{\procs}{\ensuremath{S}} 

\newcommand{\stra}{\ensuremath{\sigma}} 
\newcommand{\strb}{\ensuremath{\tau}} 
\newcommand{\opt}[1]{#1^\sharp} 
\newcommand{\straopt}{\ensuremath{\opt{\stra}}} 
\newcommand{\strbopt}{\ensuremath{\opt{\strb}}} 

\newcommand{\abs}[1]{\ensuremath{|#1|}}    
\newcommand{\norm}[1]{\ensuremath{||#1||}}    

\newcommand{\matrices}{\ensuremath{\RR^{\states\times\states}}}
\newcommand{\vectors}{\ensuremath{\RR^\states}} 

\newcommand{\m}{\ensuremath{\mathrm{Min}}}
\newcommand{\M}{\ensuremath{\mathrm{Max}}}

\newcommand{\states}{\ensuremath{\mathbf{S}}}         
\newcommand{\statesa}{\ensuremath{\states_{\M}}}         
\newcommand{\statesb}{\ensuremath{\states_{\m}}}         
\newcommand{\arn}{\ensuremath{\mathcal{A}}} 

\newcommand{\actions}{\ensuremath{\mathbf{A}}}
\newcommand{\tpseul}{\ensuremath{\delta}}
\newcommand{\tp}[2]{\delta(#1,#2)}
\newcommand{\tpstar}[2]{\delta^\star(#1,#2)}

\newcommand{\fplay}{p}
\newcommand{\ouv}[1]{\mathcal{O}(#1)}
\newcommand{\distrib}[1]{\mathcal{D}\left(#1\right)}
\newcommand{\prob}{\ensuremath{\mathbb{P}_s^{\sigma,\tau}}}
\newcommand{\expec}[4]{\ensuremath{\mathbb{E}_{#1}^{#2,#3}\left [ {#4} \right ]}}

\newcommand{\rew}{\ensuremath{r}}

\newcommand{\utlt}{\ensuremath{u}}  
\newcommand{\putlt}[1]{\ensuremath{\utlt_{#1}}} 

\DeclareMathOperator{\val}{val}

\newcommand{\ew}{\ensuremath{\mathbf{1}}} 

\newcommand{\poids}{\ensuremath{w}} 

\newcommand{\prtmap}{\ensuremath{\pi}} 

\newcommand{\priority}{\ensuremath{\pi}}

\newcommand{\RR}{\ensuremath{\mathbb{R}}}
\newcommand{\NN}{\ensuremath{\mathbb{N}}}

\newcommand{\borel}{\mathcal{B}}

\title{Blackwell-Optimal Strategies in Priority Mean-Payoff Games}
\author{Hugo Gimbert
\institute{
LaBRI, CNRS, Bordeaux, France}
\email{hugo.gimbert@labri.fr}
 \and 
Wies{\l}aw Zielonka
\institute{LIAFA, Universit\'e Paris 7 Denis Diderot, Paris, France}
\email{wieslaw.zielonka@liafa.jussieu.fr}
}

\def\titlerunning{Blackwell optimal strategies}
\def\authorrunning{H. Gimbert \& W. Zielonka}

\begin{document}
\maketitle

\begin{abstract}
We examine perfect information stochastic mean-payoff games -- a class
of games containing as special sub-classes
the usual mean-payoff games and  parity games.
We show that deterministic memoryless strategies that are 
optimal for discounted games
with state-dependent discount factors close to $1$ are
optimal for priority mean-payoff games establishing a strong link
between these two classes.
\end{abstract}

\section{Introduction}

One of the recurring themes in the  theory of stochastic games is
the interplay between discounted games and mean-payoff games.
This culminates in the seminal paper of Mertens and Neyman
\cite{mertens:neyman:81} showing that mean-payoff games have a value
and this value is the limit of the values of discounted games when the
discount factor tends to $1$. 
Note however that 
  optimal strategies in both games are very different.
As shown by Shapley~\cite{Shap:53}
discounted stochastic games admit memoryless optimal strategies.
On the other
hand mean-payoff games do not have optimal strategies, they have 
only $\varepsilon$-optimal strategies and to play optimally 
players need an unbounded memory.

The connections between discounted and mean-payoff games
become much tighter when we consider perfect information stochastic
games (games where players play in turns).
As discovered by Blackwell~\cite{blackwell:62},
if the discount factor is close to $1$ then
optimal memoryless deterministic 
strategies in discounted games 
are also optimal for
mean-payoff games (but not the other way round).  
Thus both games are related not only by their values 
but also through their optimal strategies.
Blackwell's result extends easily to two-player perfect information
stochastic games.

What happens if instead of mean-payoff games 
we consider parity games -- a class of 
games more directly
relevant to computer science~\cite{graedel:thomas:wilke:2002}?
In particular,
are parity games related to discounted games?

It is well known that deterministic mean-payoff games and parity games
are  related, see \cite{vorobyov:2004}.
The first insight that there is some link between parity games and
discounted games is due to de Alfaro
at al.~\cite{discf:2003}.
It turns out that parity games are related to multi-discounted games
with
multiple
discount factors that depend on the state.
This should be compared with discounted games with a unique, state
independent,  discount factor 
which are used in the study 
of mean-payoff games.

Like in the classical theory of stochastic games, we examine
what happens when the discount factors tend to $1$, the idea is that
in the limit we want to obtain  parity games.
Note that if we have several state dependent
discount factors $\lambda_1,\ldots \lambda_k$ then there are two
possibilities to approach $1$:
\begin{itemize}
\item
we can study the iterated limit  
$\lim_{\lambda_1\to 1}\ldots \lim_{\lambda_k\to 1}$ when discount
factors tend to $1$ one after another (i.e. first we go to $1$ with the
discount factor $\lambda_k$ associated with some group of states,
when the limit is reached then we go to $1$ with   the next discount
factor $\lambda_{k-1}$ etc.,
\item
another possibility it to examine a simultaneous limit when all factors
go to $1$ at the same time but with different rates, this will be made
precise in Section~\ref{sec:blackwell}.
\end{itemize}
The first approach is easier to handle than the second but it leads to weaker
results, 
in particular we lose the links
between optimal strategies in discounted games and optimal strategies
in parity games.

We began our examinations of
relations between discounted and parity games
in
\cite{gimbert:zielonka:2006,gimbert:zielonka:2007c}
where we limited ourselves to deterministic games.
Already this preliminary work revealed that
the natural framework for such a study  goes far beyond parity games.
In fact parity games are related to a very particular restricted class
of discounted games and when we examine all multi-discounted games then 
at the limit we obtain a new natural class of   
games --- priority mean-payoff games.
This new class contains the usual mean-payoff games and parity
games as special subclasses.

The next natural  
step is to  try to extends the results that hold for deterministic
games to perfect information stochastic games.
In two papers \cite{gimbert:zielonka:2007a,gimbert:zielonka:2007b} we
obtained some partial results in this direction.
In \cite{gimbert:zielonka:2007a}
we considered a class of games that contains parity
games but  does not contain mean-payoff games. 
We showed that such games can be seen
as an iterated limit of discounted games --- a limit in a very strong
sense, 
not only the value of the discounted games converges to the value of
the parity game but also optimal strategies in one class are inherited 
by the class of games obtained in the limit.
But these results are not satisfactory for two reasons, the class of games 
for which we were able to carry our study is too restrictive.
This class 
involves some technical restrictions on discounted games, 
which are natural for parity games, but 
not so natural for discounted games. 
The second problem comes from the fact that
\cite{gimbert:zielonka:2007a} uses
the iterated limit of discount factors and not the more interesting
simultaneous limit.


In the second paper \cite{gimbert:zielonka:2007b} we considered
priority mean-payoff games in full generality, with no artificial
restrictions,  and we examined directly the limit
with the discount factors tending to $1$ with different rates rather
than the iterated limit. However \cite{gimbert:zielonka:2007b}
deals  only with one-player games and it examines only games values,
the paper does not provide any
relation between optimal strategies in multi-discounted games and
optimal strategies in the priority mean-payoff games in the limit.

In the present paper we remove all restrictions imposed in 
\cite{gimbert:zielonka:2007a,gimbert:zielonka:2007b}. We consider
the full class perfect information stochastic priority mean-payoff games and we
show that such games are a limit of discounted games with discount
factors tending to $1$ with the rates depending on the priority.
Not only at the limit the value of the discounted game equals to the
value of the priority mean-payoff game but also optimal 
deterministic memoryless strategies in discounted games turn out to be
optimal in the the corresponding priority mean-payoff game.

The interest in such a result is threefold. 

First we think that 
establishing a very strong link between two apparently different
classes of games has its own intrinsic interest.  

Discounted
games were thoroughly studied in the past and  our result shows that
algorithms for such games can, in principle, be used to solve parity
games (admittedly all depends on how much the discount factor should be
close to $1$ in order that two types of games become close enough,
and this remains open). 

Another point concerns the stability of
solutions (optimal strategies and games values) under small
perturbations.
When we examine stochastic games then the natural question is where
the transition probabilities come from? If they come from an
observation then the values of transition probabilities are not
exact. On the other hand algorithms for stochastic games use only
rational transition probabilities thus even if we know the exact
probabilities we replace them by close rational values. What 
is the impact of such approximations on solutions, are optimal
strategies stable under small perturbations?
Usually we tacitly assume that this is the case but it would be better
to be sure. Since Blackwell-optimal strategies studied in
Section~\ref{sec:blackwell} are
stable under small perturbations of discount factors (because they do
not depend on the discount factor) this adds some credibility to the
claim that
Blackwell optimal strategies are stable for parity games.

And the last point. Blackwell invented Blackwell optimality because he
was not satisfied with the notion of optimal strategies for
mean-payoff Markov decision processes. 
However the same can be said about parity games, we defer  examples
to the final section.
 
The paper is organized as follows.
In Section~\ref{sec:introgames} we introduce stochastic games in
general, we define the notions of value and optimal strategies.
Section~\ref{sec:discounted} we examine discounted games.
The main result in this section shows that if discount factors
 are close to
$1$ then optimal strategies stabilize (Blackwell optimality).
In Section~\ref{sec:priority} we introduce the class of priority
mean-payoff games --- this is the principal class of games examined in 
this paper. Parity games and mean-payoff games are just very special
subclasses of this class.
In Section~\ref{sec:fromblackwell} we prove the main result of the
paper stating that deterministic memoryless strategies optimal for
discounted games for discount factors sufficiently close to $1$ 
 are optimal in derived priority mean-payoff games.


\section{Stochastic Games with Perfect Information}
\label{sec:introgames}
\paragraph{Notation.} 
In this paper $\NN$ stands for the set of positive integers,
$\NN_0=\NN\cup\{0\}$, and $\RR_+$ is the set of positive real numbers.

For each finite set $X$, $\distrib{X}$ is the set of probability
distributions over $X$, i.e. it is the set of mappings $p : X\to
[0,1]$ such that $\sum_{x\in X}p(x)=1$. The support of
$p\in\distrib{X}$ is the set $\{x\in X : p(x)>0\}$.
\subsection{Games and Arenas}
Two players $\M$ and $\m$ are playing an infinite game on an arena.
An arena  is a tuple 
\[
\arn=(\states,\statesa,\statesb,\actions,(\actions(s))_{s\in\states},\tpseul),
\]
where a finite set of states $\states$ is partitioned in two sets,
the set
$\statesa$ of states 
controlled by player $\M$ and
the set $\statesb$ of states controlled by player $\m$.
For each state $s\in \states$ there is a non-empty finite set $\actions(s)$ of
actions available in $s$, $\actions=\bigcup_{s\in\states} \actions(s)$.
Players \M\ and \m\ play on \arn\ an infinite game.
If at stage $i\in\NN_0$ the game is in a state $s_i\in\states$
then the player controlling  $s$ chooses an action from  
$\actions(s)$ and a new state $s_{i+1}$ is chosen
with probability 
specified by the transition mapping $\tpseul$.
Transition mapping \tpseul\ maps each pair $(s,a)$, where
$s\in\states$
and $a\in\actions(s)$, to an element of
$\distrib{\states}$. Intuitively, if  in a
 state $s$ and  an action $a$ is executed then 
$\tp{s}{a}(t)$ gives the probability
that at the next stage the game  is in  state $t$.
To simplify the notation we shall 
write $\tp{s,a}{t}$  rather than $\tp{s}{a}(t)$.

Throughout the paper we assume that all arenas are finite, i.e. 
the sets of states and actions are finite.

An arena is said to be a \emph{one-player arena} controlled by player $\M$ 
if, for every state $s$ controlled by $\m$, 
the set  $\actions(s)$ is a singleton (in particular if all states are
controlled by \M\ then \arn\ is a one-player arena controlled by \M).
One-player arenas controlled by player $\m$ are defined similarly.

A finite (resp. infinite) \emph{play} in the arena $\arn$ is a non-empty 
finite (resp. infinite) sequence of states and actions
in $(\states\actions)^*\states$ (resp. in $(\states\actions)^\omega$).
In the sequel ``play'' without any attribute will be used as a synonym of
``infinite play''.

\subsection{Payoffs}

After an infinite play player \M\ receives a payoff from 
player \m. The objectives of the players are opposite,
the goal of \M\ is to maximize the payoff while player \m\ wants to
minimize the payoff.

The payoff can be computed in various ways. For example in a mean-payoff
game each state is labeled with a real number called the reward
and after an infinite play the payoff of player \M\ is the limit of mean
values of the 
sequence of rewards. 
In a parity game,
each state is labeled with an integer called a priority and player
\M\ receives payoff 
$0$ or $1$ depending on the parity of the highest priority seen
infinitely often. 
In both examples, the way the payoffs are computed is independent
from the transitions rules of the game (the arena), it depends
uniquely on the play.

%


Thus formally a payoff function is a mapping
\[
\utlt : (\states\actions)^\omega \to \RR
\]
from infinite plays to real numbers.

A game is a couple $\Gamma=(\arn,\utlt)$ made of an arena and a payoff function.
Usually we consider not a particular game but rather a class of
games. In this case arenas are endowed with some additional structure, usually
some labeling of states or actions (for example
rewards as in mean-payoff games or priorities as in parity games) and 
this labeling is used to define the payoff for games in the given class.

\subsection{Strategies}

Playing  a game the players use strategies.
A \emph{strategy} for player $\M$ is a mapping 
$\stra:(\states\actions)^*\statesa\to \distrib{\actions}$
such that for every finite play $p=s_0a_0s_1a_1\ldots s_n$ with $s_n\in\statesa$,
the support of $\stra(p)$ is a subset of the actions available in $s_n$,
i.e. for all $a \in\actions$, 
if $\stra(p)(a)>0$ then $a\in\actions(s_n)$.

Strategies for player $\m$ are defined similarly and denoted $\strb$.

Certain types of strategies are of particular interest.
A strategy is \emph{deterministic}
if it chooses actions in a deterministic way,
and it is \emph{memoryless} if
it does not have any memory,
i.e.  choices depend only on the current state of the game, and not on
the past history. Formally:

\begin{definition}\label{defi:strats}
A strategy $\stra$ of player $i\in\{\m,\M\}$ is said to be:
\begin{itemize}
\item[\textbullet] \emph{deterministic} 
if, $\forall \fplay\in(\states\actions)^*\states_i$, 
if $\stra(\fplay)(a) > 0$ then $\stra(\fplay)(a)=1$,
\item[\textbullet] \emph{memoryless} 
if, $\forall t\in\states_i$ and $\fplay\in(\states\actions)^*$, 
$\stra(\fplay t)=\stra(t)$.
\end{itemize}
\end{definition}



For any finite play $\fplay\in(\states\actions)^*\states$ and
an action $a\in\actions$ we define
the cones $\ouv{\fplay}$ and $\ouv{\fplay a}$ 
as the sets consisting  
of all infinite plays with prefix \fplay\ and $\fplay a$ respectively. 

In the sequel we assume that 
the set of infinite plays $(\states\actions)^\omega$ 
is equipped with the $\sigma$-field $\borel((\states\actions)^\omega)$
generated by the collection of all cones $\ouv{\fplay}$ and
$\ouv{\fplay a}$.
Elements of this $\sigma$-field are called  \emph{events}.
Moreover, when there is no risk of confusion,
the events $\ouv{\fplay}$
and $\ouv{\fplay a}$ will be denoted simply $\fplay$ and $\fplay a$.

Suppose that players \M\ and \m\ 
are playing accordingly to strategies $\stra$ and $\strb$.
Then after a finite play $s_0a_1\ldots s_n$ the probability of
choosing an actions $a_{n+1}$ is either 
$\stra(s_0a_1\ldots s_n)(a_{n+1})$
or 
$\strb(s_0a_1\ldots s_n)(a_{n+1})$ depending on whether $s_n$ belongs
to $\statesa$ or to $\statesb$.
Fixing the initial state $s\in\states$ these probabilities and
the transition probability \tpseul\ yield
the following probabilities
\begin{equation}
\label{eq:decal0}
\prob( s_0) =\begin{cases}
             1 & \text{if $s_0=s$} \\
             0 & \text{if $s_0\neq s$}
\end{cases}
\end{equation}
is the probability of the cone $\ouv{s_0}$,
\begin{equation}
\label{eq:decal1}
\prob( s_0a_1\ldots s_na_{n+1} \mid s_0a_1\ldots s_n) =
  \begin{cases}
     \stra(s_0a_1\ldots s_n )(a_{n+1}) & \text{if $s_n\in \statesa$} \\
     \strb(s_0a_1\ldots s_n )(a_{n+1}) & \text{if $s_n\in \statesb$} 
  \end{cases}
\end{equation}
is the conditional probability of $\ouv{s_0a_1\ldots s_n a_{n+1}}$
given $\ouv{s_0a_1\ldots s_n}$ and
\begin{equation}
\label{eq:decal2}
\prob( s_0a_1\ldots s_n a_{n+1}s_{n+1} \mid s_0a_1\ldots s_n a_{n+1})
=\tp{s_n,a_{n+1}}{s_{n+1}}
\end{equation}
is the conditional probability of the cone 
$\ouv{s_0a_1\ldots s_n a_{n+1}s_{n+1}}$
given the cone $\ouv{s_0a_1\ldots s_n a_{n+1}}$.

Ionescu Tulcea's theorem~\cite{shiryayev}
implies that
there exists a unique probability measure $\prob$ on the measurable space
$((\states\actions)^\omega,\borel(\states\actions)^\omega)$
satisfying \eqref{eq:decal0}, \eqref{eq:decal1} and \eqref{eq:decal2}.

\subsection{Optimal strategies}

Let
$\arn=(\states,
\statesa,\statesb,\actions,(\actions(s))_{s\in\states},\tpseul)$
be an arena. In the sequel we assume that all payoff mappings $\utlt :
(\states\actions)^\omega \to \RR$ are bounded and measurable
(for measurability we assume that $(\states\actions)^\omega$ 
is equipped with the $\sigma$-field 
described in the preceding section and \RR\ is equipped with the $\sigma$-field
$\borel(\RR)$ of Borel sets).

Given 
an initial state $s$ and  strategies $\stra$ and $\strb$ of \M\ and \m\
the expected value of the payoff $\utlt$ under $\prob$ is denoted
$\expec{s}{\stra}{\strb}{\utlt}$.

A strategy $\straopt$ for player $\M$ 
is said to be $\emph{optimal}$ in a game $(\arn,\utlt)$
if for every state $s$,
\[
\inf_\strb \expec{s}{\straopt}{\strb}{\utlt} = 
\sup_\stra  \inf_\strb \expec{s}{\stra}{\strb}{\utlt}\enspace.
\]
Dually a strategy $\strbopt$ of player \m\ is \emph{optimal}
if 
$\sup_\stra \expec{s}{\stra}{\strbopt}{\utlt} = 
\inf_\strb \sup_\stra   \expec{s}{\stra}{\strb}{\utlt}$, for each state $s$.


In general, 
\[
\underline{\val}_s(\utlt):=\sup_\stra  \inf_\strb \expec{s}{\stra}{\strb}{\utlt} \leq
\inf_\strb \sup_\stra \expec{s}{\stra}{\strb}{\utlt} :=\overline{\val}_s(\utlt)
\]
but when these two quantities are equal
then 
the state $s$ is said to have the \emph{value} 
$\val_s(\utlt)=\underline{\val}_s(\utlt)=\overline{\val}_s(\utlt)$, denoted also
$\val_s(\utlt,\arn)$ whenever mentioning explicitly the arena  is needed.
Under the hypothesis that $\utlt$ is measurable and bounded,
Martin's theorem~\cite{Martin:98}
guarantees that every state has a value.
Notice however that Martin's theorem does not guarantee the existence 
of optimal strategies.

\section{Discounted Games}
\label{sec:discounted}
Arenas for discounted games
are equipped with two mappings
defined on the set \states\ of states.
The \emph{discount mapping} 
\begin{equation*}
\lambda : \states \lra [0,1)
\end{equation*}
 associates
with each state $s$ a discount factor $\lambda(s)\in[0,1)$ 
and
the \emph{reward mapping} 
\begin{equation}
\rew : \states \lra \rr
\label{eq:rewrd}
\end{equation}
maps each state $s$ to a real valued reward $\rew(s)$.

The payoff 
\[
\putlt{\lambda} : (\states\actions)^\omega \lra \rr
\]
for discounted games is calculated in the following way.
For each play $p=s_0a_0s_1a_1s_2a_2\ldots\in(\states\actions)^\omega$
\begin{align}
\notag \putlt{\lambda}( p ) &= 
(1-\lambda(s_0))\rew(s_0)  +\lambda(s_0)(1-\lambda(s_1))\rew(s_1) + 
\lambda(s_0)\lambda(s_1)(1-\lambda(s_2))\rew(s_2) + \ldots \\
&=\sum_{i=0}^\infty
\lambda(s_0)\ldots\lambda(s_{i-1})(1 - \lambda(s_i))\rew(s_i) \enspace .
\label{eq:defmu}
\end{align}

Usually when discounted games are considered 
it is assumed that there is only one discount
factor, i.e. that there exists $\lambda\in [0,1)$ such that
$\lambda(s) = \lambda$ for all $s\in\states$.
But for us it is essential that the discount factor depends on the state.

Shapley~\cite{Shap:53} proved\footnote{In fact, Shapley
considered a much larger class of stochastic games. For these games he
proved that both players have memoryless optimal strategies. For
perfect information  games   his proof
yields optimal strategies that are also deterministic.}
that
\begin{theorem}[Shapley]
Discounted games 
$(\arn, \putlt{\lambda})$ over finite arenas admit optimal
deterministic memoryless strategies
for both players.
\label{th:shapley}
\end{theorem}

\subsection{Interpretations of discounted games}
\label{sec:interpretations}
The rather obscure formula \ref{eq:defmu} can be interpreted in
several ways.
The usual economic interpretation is the following. The reward
$\rew(s)$ represents the payoff that player \M\ receives if the state
$s$ is visited. But a given sum of  money is worth more now than in
the future, visiting $s_i$ at stage $i$ is worth 
$\lambda(s_1)\ldots\lambda(s_{i-1})\rew(s_i)$ rather than $\rew(s_i)$
(visiting $s_i$ is worth $\rew(s_i)$ only the first day).
With this interpretation 
$\sum_{i=0}^\infty \lambda(s_0)\ldots\lambda(s_{i-1})\rew(s_i)$ represents
the accumulated total the payoff that player \M\ receives during an infinite
play. 
However, with this interpretation 
it is difficult to assign a meaning to the factors $(1-\lambda(s_i))$
and such factors are essential when we consider the limit of
$\putlt{\lambda}$ with discount factors tending to $1$.

In his seminal paper \cite{Shap:53} Shapley
gives  another interpretation of 
\eqref{eq:defmu} in terms
\emph{stopping games}.
Suppose that at a stage $i$ 
a state  $s_i$ is visited. 
Then with
probability $1-\lambda(s_i)$ the nature can stop the game.
Since we have assumed that $0\leq \lambda(s)<1$
for all $s \in\states$,
the stopping probabilities are strictly 
positive which implies that the game will   eventually stop with probability $1$
after a finite number of steps.

If the game stops in $s_i$ then player
\M\ receives from player \m\ the payment $\rew(s_i)$
and
this ends the game. Thus here player \M\ receives the payoff only
once, when the game stops and the payoff is determined by the last
 state. 

If the game does not stop in $s_i$
then there is no payment at this stage and the player controlling the
state $s_i$
chooses an action to execute.

Note that 
$\lambda(s_0)\ldots\lambda(s_{i-1})(1-\lambda(s_i))$ gives the probability
that the game has not stopped in any of the states
 $s_0,\ldots,s_{i-1}$ but it does
stop in the state $s_i$. Since this
event   results in the payment $\rew(s_i)$, 
\eqref{eq:defmu} represents in this interpretation the
\emph{payoff  expectation} for an infinite play
$s_0a_0s_1a_1s_2a_2\ldots$ during the stopping game.

Another related interpretation making a direct link between discounted
games and mean-payoff games is the following. We transform the
discounted arena
\arn\  into a new arena $\arn^\star$ 
by attaching to each state $s\in\states$ a new state $s^\star$.
We set $\rew(s^\star)=\rew(s)$, i.e. each new adjoined state has the
same reward as the corresponding original state.

In the new arena $\arn^\star$ we incorporate the discount factors directly into the
transition probabilities.
Recall that, for each state $s\in\states$ of the original arena \arn,
$\tp{s,a}{s'}$ was the probability 
of going to a state $s'$ if an action $a$ is executed in $s$. In the
new arena $\arn^\star$ this probability is set to
$\tpstar{s,a}{s'}=\lambda(s)\tp{s,a}{s'}$.
On the other hand we set also
$\tpstar{s,a}{s^\star} = (1-\lambda(s))$, i.e. in 
$\arn^\star$ with probability $1-\lambda(s)$
the execution of $a$ in $s$ leads to $s^\star$ (note that for fixed
$a$ the probabilities sum  up to $1$).

Each new state $s^\star$ is absorbing,
there is only one  action available in each $s^\star$, we
note it $\star$, and this action leads with probability 
$1$ back to $s^\star$.
This situation is illustrated by the following picture.
\begin{center}
\includegraphics[scale=0.9]{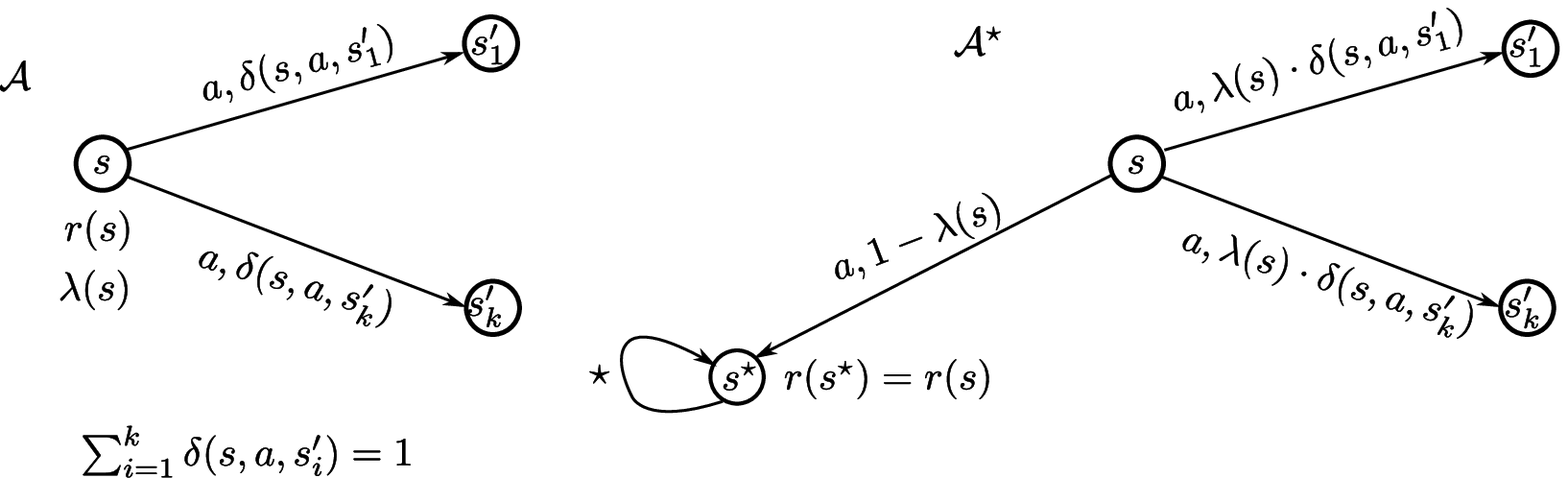}
\end{center}

We consider the mean-payoff game played on $\arn^\star$, i.e. the game
with the payoff $\putlt{\rew}(s_0a_0s_1a_1\ldots)=\limsup_{k}\frac{1}{k+1}
\sum_{i=0}^k \rew(s_i)$. Such a game played on $\arn^\star$ ends with
probability $1$ in one of the starred states $s^\star$ 
and then the mean-payoff
is simply $\rew(s^\star)=\rew(s)$. Intuitively, stopping in $s$ 
with the payoff $\rew(s)$ in the stopping game
 is the same as going to $s^\star$ and
looping there infinitely with the same mean-payoff $\rew(s^\star)$.  
Thus a discounted game can be seen as a  mean-payoff game played on an
arena where with probability $1$ we end in some absorbing state.
If discount factors tend to $1$ then  this means that, intuitively, we
cut off the absorbing starred states of $\arn^\star$.
\section{Blackwell optimality}
\label{sec:blackwell}
We will consider what happens if the discount factors tend to $1$.
The novelty in comparison with the traditional approach is that we
consider 
the situation where discount factors of different states tend to $1$ with
different rates.

A \emph{rational discount parametrization} is a family of mappings 
$\lambda_t=(\lambda_t(s))_{s\in\states}$,
such that for each state $s$,
\begin{itemize}
\item
$t\mapsto \lambda_t(s)$ is a 
rational\footnote{Rational in the sense that $\lambda_t(s)$ 
is a quotient of two polynomials of $t$.}
mapping of $t$,  
\item
there exists $0<\varepsilon<1$ such that $\lambda_t(s)\in[0,1)$ for
all $t\in[1- \varepsilon,1)$ (note that since the set of states is
finite
we can choose the same $\varepsilon$ for all states),
\item
$\lim_{t\uparrow 1}\lambda_t(s) = 1$.
\end{itemize}

A typical example of a rational parametrization is the \emph{canonical
  rational discount parametrization} defined in the following way.
For each state $s$ we fix a natural number $\priority(s)\in\nn$ called the
priority of $s$ and a positive real number $\poids(s)\in(0,\infty)$ called
the weight of $s$. Then the canonical parametrization is defined as
\begin{equation}
 \lambda_t(s)=1-\poids(s)(1-t)^{\priority(s)}, \quad 
\text{for $s\in\states, t\in\RR$.} 
\label{eq:canonical}
\end{equation}

We will consider discounted games where discount factors are given by
a rational discount parametrization.

\begin{theorem}[Blackwell optimality]
\label{th:blackwell}
Let us fix an arena \arn\ of a discounted game and
let $\lambda_t$ be a rational discount parametrization for \arn.
Let $\val_s(\putlt{\lambda_t})$ be the value of a state $s\in\states$ for 
$\lambda_t$ in the game $(\arn,\putlt{\lambda_t})$.

Then  there exists $0<\varepsilon<1$ such that, for each state $s$,
\begin{enumerate}[(1)]
\item
for $t\in(1-\varepsilon,1)$, 
$t \mapsto \val_s(\putlt{\lambda_t})$ is a rational function of $t$ and
\item
if \straopt\ and \strbopt\ are optimal deterministic memoryless
strategies for some 
$t\in(1-\varepsilon,1)$ then \straopt\ and \strbopt\ are optimal for all
$t\in(1-\varepsilon,1)$.
\end{enumerate}
\end{theorem}

In the sequel we call strategies \straopt\ and \strbopt\ Blackwell
optimal for a rational discount parametrization
 $\lambda_t$ if  \straopt\ and \strbopt\ are deterministic
memoryless strategies satisfying part (2) of
Theorem~\ref{th:blackwell}.

Let us note that Theorem~\ref{th:blackwell} exhibits a curious
property of discounted games discovered by
Blackwell~\cite{blackwell:62}\footnote{In fact 
Blackwell~\cite{blackwell:62} considered
only one-player games with the same discount factor for all states.}.
By Theorem~\ref{th:shapley} we know that
for each fixed $t$ the discounted game with payoff $\putlt{\lambda_t}$
has optimal memoryless deterministic strategies, but obviously such
strategies depend on $t$. 
Theorem~\ref{th:blackwell} asserts that for $t\in(1-\epsilon,1)$ 
the situation stabilizes and optimal deterministic memoryless
strategies do not depend on $t$.
Since Blackwell optimality is usually proved 
only for Markov decision processes
 with a unique discount factor for all states, see \cite{hor:yush:2002} for
example,
we decided to include the complete proof of Theorem~\ref{th:blackwell}.
Note however that our proof follows closely the one
used for Markov decision processes.

The proof of Theorem~\ref{th:blackwell} is based on the following
lemma that will be useful also in the next section.

\begin{lemma}\label{lem:rational}
Let $t\mapsto \lambda_t$ be a rational discount parametrization and
let $\stra,\strb$ 
be deterministic memoryless strategies.
 Then, for each state $s$, and for $t$ sufficiently close to $1$,
$\expec{s}{\stra}{\strb}{\utlt_{\lambda_t}}$
is a rational function of $t$.
\end{lemma}
\begin{proof}
The proof is standard but we give it for the sake of completeness.
The set \matrices\  of functions from
$\states\times\states$ into real numbers can be seen as the set of
square real valued matrices with rows and columns indexed by \states.
In particular \matrices\ is a vector space with natural matrix
addition and scalar multiplication.
However, matrix multiplication defines also a product  on \matrices,
for $M,N\in\matrices$, $MN$ is an element $U$ of \matrices\ with entries
$U[s',s'']=\sum_{s\in\states}M[s',s]N[s,s'']$.
We endow \matrices\ with a norm, for $M\in\rr^{\states\times\states}$,
$\norm{M}=\max_{s'\in\states}\sum_{s''\in\states}\abs{M[s',s'']}$.
It can be easily shown that $\norm{MN}\leq\norm{M}\cdot\norm{N}$ for
$M,N\in\matrices$ and
\matrices\ is a complete metric space for
the metric induced by the norm $\norm{\cdot}$, 
see~Section~3.2.1 of \cite{stroock}
 for
a proof.

On the other hand, we consider also the vector space \vectors\ of functions from
\states\ into \rr, they can be seen as column vectors indexed by
states. Of course if $M\in\matrices$ and $v\in\vectors$ then
$Mv\in\vectors$, where $(Mv)[s]=\sum_{s'\in\states}M[s,s']v[s']$ for
$s\in\states$. 

We equip \vectors\ with a norm, for $v\in\vectors$,
$\norm{v}_\infty=\max_{s\in\states}\abs{v[s]}$.
The norms on \matrices\ and \vectors\ are compatible in the sense that
we have
$\norm{Mv}_\infty\leq\norm{M}\cdot\norm{v}_\infty$.

Let $\sigma,\tau$ be  deterministic memoryless strategies for players
\M\ and \m\ and let $\lambda_t$ be a rational discount
parametrization.
We define 
\[
\tp{s'}{s''}= \begin{cases}
               \tp{s',\sigma(s')}{s''} & \text{if $s'\in\statesa$,}\\
               \tp{s',\tau(s')}{s''} & \text{if $s'\in\statesa$,}
\end{cases} \quad \text{for $s',s''\in\states$.}
\]
Thus $\tpseul$ 
defines transition probabilities of the Markov chain obtained when we
fix the strategies $\sigma$ and $\tau$.

In the sequel $M$ will denote the element of \matrices\ defined in the
following way
\begin{equation}
M[s',s'']=\lambda_{t}(s')\tp{s'}{s''},
\quad \text{for $s',s'',\in\states$}
\label{def:matricem}
\end{equation}

Let $I\in\matrices$ be the identity matrix, i.e. $I[s',s'']$ is $1$ if
$s'=s''$ and $0$ otherwise.

We shall show that for $t$ close to $1$ 
the matrix $(I-M)$ is invertible and
\begin{equation}
(I-M)^{-1}
=\sum_{i=0}^\infty M^i .
\label{eq:inverse}
\end{equation}

First we show that the series on the right-hand side of
\eqref{eq:inverse}
converges.

Let $\lambda_{M} = \max_{s\in\states}\lambda_t(s)$. Then for $t$
sufficiently close to $1$ we have
$\norm{M}
\leq \lambda_M < 1$ and,
for $k < l$,
\[ \norm{\sum_{i=k}^l M^i} \leq
\sum_{i=k}^l 
\norm{M}^i \leq
\sum_{i=k}^l (\lambda_M)^i = \frac{\lambda_M^k -
  \lambda_M^{l+1}}{1-\lambda_M} \xrightarrow[ k,l\rightarrow\infty ]{} 0
\] 
since, by the definition of a rational discount parametrization,
 $0\leq \lambda_M < 1$ for $t$ sufficiently close to $1$. 
Thus the series $\sum_{i=0}^\infty M^i$ satisfies the Cauchy
condition and the convergence follows from the completeness of the norm
$\norm{\cdot}$.
Now it suffices to note that 
\[
(I-M)^{-1}\cdot\sum_{i=0}^k
  M^i - I =
  M^{k+1}
\] 
and 
$\norm{M^{k+1}}\leq \norm{M}^{k+1} \leq
\lambda_M^{k+1}\xrightarrow[k\rightarrow\infty]{} 0$ 
which yields \eqref{eq:inverse}.

Let $(\procs_i)_{i=0}^\infty$ be the stochastic process giving the
state at stage $i$.
Then
\begin{multline}
\expec{s}{\sigma}{\tau}{\putlt{\lambda_t}}=
\expec{s}{\sigma}{\tau}{\sum_{i=0}^\infty \lambda_t(\procs_0)\cdots
\lambda_{t}(\procs_{i-1})(1-\lambda_{t}(\procs_i))\rew(\procs_i)}\\
=\lim_{k\rightarrow \infty}
\expec{s}{\sigma}{\tau}{\sum_{i=0}^k \lambda_{t}(\procs_0)\cdots
\lambda_{t}(\procs_{i-1})(1-\lambda_{t}(\procs_i))\rew(\procs_i)}
\label{eq:lebesgue}
\end{multline}
where the second equality follows from the Lebesgue dominated
convergence theorem.

Let $v$ 
be an element of \vectors\ defined as
\[
v[s]=(1-\lambda_t(s))\rew(s), \quad \text{for
  $s\in\states$}.
\]
An elementary induction on $i$ shows that, for $s,s'\in\states$,
\[
\expec{s}{\sigma}{\tau}{ 
\lambda_{t}(\procs_0)\cdots \lambda_{t}(\procs_{i-1}) |
\procs_0=s,\procs_i=s'}
=M^i[s,s'],
\]
i.e. the entry $[s,s']$ of the $i$-th power of
$M$ is the expectation of $\lambda_{t}(\procs_0)\cdots
\lambda_{t}(\procs_{i-1})$
under the condition that $\procs_0=s$ and $\procs_i=s'$.
This yields
\begin{multline}
(M^i v)[s] =
\sum_{s'\in\states}M^i[s,s']\cdot
v[s']=
\sum_{s'\in\states}
\expec{s}{\sigma}{\tau}{ \lambda_{t}(\procs_0)\cdots
\lambda_{t}(\procs_{i-1}) |
\procs_0=s,\procs_i=s'}\cdot (1-\lambda_t(s'))\rew(s') =\\
\expec{s}{\sigma}{\tau}{ \lambda_{t}(\procs_0)\cdots
\lambda_{t}(\procs_{i-1})(1-\lambda_t(\procs_i))\rew(\procs_i) |
\procs_0=s}.
\label{eq:zorro}
\end{multline}
Taking the sum from $i=0$ to $k$ on
both sides of \eqref{eq:zorro} and next the limit with $k$ tending to
infinity,
using \eqref{eq:lebesgue} and
\eqref{eq:inverse},
we obtain
\[
\expec{s}{\sigma}{\tau}{\putlt{\lambda_t}}=
((I-M)^{-1}v)[s] .
\]

But the elements of the matrix $I-M$ are rational functions of $t$,
thus Cramer's rule for matrix inversion show that $(I-M)^{-1}$ has
also rational elements, and since the elements of $v$ are also
rational functions we can see that
$\expec{s}{\sigma}{\tau}{\putlt{\lambda_t}}$ is a rational function of $t$.

\end{proof}

\begin{proof}[Proof of Theorem \ref{th:blackwell}]
According to Lemma~\ref{lem:rational},
and since discounted games admit optimal deterministic memoryless strategies,
(1) is a consequence of (2).

We prove (2)  as follows.

Let $X$ be the set of all tuples $(q,\sigma,\tau,\sigma',\tau')$,
where $q$ is a state, $\sigma, \sigma'$ are deterministic memoryless
strategies for player \M\ and $\tau, \tau'$ are deterministic memoryless
strategies for player \m.  Note that for finite arenas $X$ is finite.
Let $\lambda_t$ be a rational discount parametrization and let
$0<\varepsilon<1$
be such that $\lambda_t(s)\in(0,1)$ for all states $s$ and all
$t\in(1-\varepsilon,1)$. 

For each $(q,\sigma,\tau,\sigma',\tau')\in X$ we consider the function
$\Phi_{q,\sigma,\tau,\sigma',\tau'} : (1-\varepsilon,1)\to\rr$ defined by: 
\[
t \mapsto
\Phi_{q,\sigma,\tau,\sigma',\tau'}(t) =
\expec{q}{\stra}{\strb}{\utlt_{\lambda(t)}}-
\expec{q}{\stra'}{\strb'}{\utlt_{\lambda(t)}}\enspace. 
\]
According to Lemma~\ref{lem:rational},
$\Phi_{q,\sigma,\tau,\sigma',\tau'}(t)$ is a rational function of $t$
for $t$ sufficiently close to $1$. Since a rational function can
change the sign (cross the $x$-axis) only finitely many times
there exists $\varepsilon_1=\varepsilon_1(q,\sigma,\tau,\sigma',\tau')>0$ 
such that the sign of
$\Phi_{q,\sigma,\tau,\sigma',\tau'}(t)$ does not change in the
interval 
$(1-\varepsilon_1,1)$. Let 
$\varepsilon_2=\min\{\varepsilon\}\cup
\{\varepsilon_1(q,\sigma,\tau,\sigma',\tau'):(q,\sigma,\tau,\sigma',\tau')\in
X\}$.

Since $X$ is finite the minimum on the right is taken over a finite
set of positive numbers and
we conclude that $\varepsilon_2>0$

Let us take any $t\in(1-\varepsilon_2,1)$.
Let $\straopt$, $\strbopt$ be optimal deterministic memoryless
strategies in the discounted game $(\arn,\putlt{\lambda_t})$
(Theorem~\ref{th:shapley}). 
Then, in particular, we have
\begin{equation}
\expec{q}{\sigma}{\strbopt}{\putlt{\lambda_t}} \leq
\expec{q}{\straopt}{\strbopt}{\putlt{\lambda_t}} \leq
\expec{q}{\straopt}{\tau}{\putlt{\lambda_t}}  
\label{eq:eez}
\end{equation}
for all deterministic memoryless strategies $\sigma,\tau$.
We can rewrite \eqref{eq:eez} as
$\Phi_{q,\straopt,\strbopt,\sigma,\strbopt}(t)\geq 0$ and\\
$\Phi_{q,\straopt,\tau,\straopt,\strbopt}(t)\geq 0$. 
However if these inequalities hold for some $t\in(1-\varepsilon_2,1)$
then we have seen that they hold for all $t\in(1-\varepsilon_2,1)$.
Therefore \eqref{eq:eez} holds for all $t\in(1-\varepsilon_2,1)$.
Finally Theorem~\ref{th:shapley} implies that
if \eqref{eq:eez} holds for all deterministic memoryless strategies
$\sigma$ and $\tau$ 
(with fixed deterministic memoryless $\straopt$ and $\strbopt$)
then it holds for all strategies $\sigma, \tau$\footnote{In other
  words, for discounted games 
being optimal in the class of memoryless deterministic
  strategies implies being optimal in the class of all strategies.}.
\end{proof}

\section{Priority mean-payoff games}
\label{sec:priority}
In mean-payoff games the players try to optimize (maximize/minimize)
the mean value of the payoff received at each stage.
In such games  the \emph{reward mapping} 
\begin{equation}
\rew : \states \lra \rr
\label{eq:rewardmap}
\end{equation}
gives, for each state $s$,   the payoff received by
player \M\ when $s$ is visited.
The payoff of an infinite
play is defined as the limit of the means  of daily payments:
\begin{equation}
\putlt{\rew}(s_0s_1s_2\ldots) = \limsup_{k} \frac{1}{k+1} \sum_{i=0}^k
  \rew(s_i) \enspace ,
\label{eq:simplemean}
\end{equation}
where we take $\limsup$ rather than the simple limit 
since the latter may not exist.

We slightly generalize  mean-payoff games by 
equipping arenas with a new mapping
\begin{equation}
\poids :   \states \lra \rr_+
\label{eq:poids}
\end{equation}
associating with each state $s$ a \emph{strictly positive} real number $\poids(s)$, the
\emph{weight} of $s$. We can interpret $\poids(s)$ as the amount of
time spent in  state $s$  upon each visit to $s$. In this
setting $\rew(s)$ should be seen as the payoff by a time unit when $s$
is visited, 
thus the weighted mean payoff received by player \M\ is
\begin{equation}
\putlt{\rew,\poids}(s_0s_1s_2\ldots) = \limsup_{k} \frac{\sum_{i=0}^k
  \poids(s_i)\rew(s_i)}{\sum_{i=0}^k \poids(s_i)}\enspace .
\label{eq:notsimplemean}
\end{equation}
Note that in the special case when the weights are all equal to $1$,
the weighted mean value \eqref{eq:notsimplemean} reduces to 
\eqref{eq:simplemean}.

As a final ingredient we add to the arena 
a \emph{priority mapping}
\begin{equation}
\prtmap : \states \lra \NN
\label{eq:prtmap}
\end{equation}
assigning to each state $s$
a positive integer \emph{priority} $\prtmap(s)$.

We define  the \emph{priority} of a play  $p=s_0a_0s_1a_1s_2a_2\ldots$ as the
\emph{smallest} priority appearing infinitely often in the sequence
$\prtmap(s_0)\prtmap(s_1)\prtmap(s_2)\ldots$
of  priorities visited in $p$:
\begin{equation}
 \priority(p) = \liminf_i  \prtmap(s_i)\enspace .
\label{eq:prty}
\end{equation}
For any priority $\alpha$, let
$\ew_\alpha : \states \lra \{0,1\}$ be the indicator function
of the set $\{ s\in\states\mid \prtmap(s) = \alpha \}$, i.e.
\begin{equation}
 \ew_\alpha(s)=\begin{cases}
1 & \text{if $\prtmap(s) = \alpha $} \\
0 & \text{otherwise}.
\end{cases}
\label{def:indic}
\end{equation}

Then the priority mean-payoff of a play $p=s_0a_0s_1a_1s_2a_2\ldots$ is defined as
\begin{equation}
\putlt{\rew,\poids,\prtmap}(p) = \limsup_k 
\frac{\sum_{i=0}^k \ew_{\priority(p)}(s_i) \cdot \poids(s_i) \cdot \rew(s_i)}%
{\sum_{i=0}^k  \ew_{\priority(p)}(s_i)\cdot \poids(s_i)} \enspace .
\label{eq:mpay}
\end{equation}

In other words, to calculate priority mean payoff
$\putlt{\rew,\poids,\priority}(p)$ 
we take weighted mean payoff but with the weights of all
states
having priorities different from $\priority(p)$
shrunk to $0$. (Let us note that
the denominator $\sum_{i=0}^k  \ew_{\priority(p)}(s_i)\cdot
\poids(s_i)$ is different from $0$ for $k$ large enough, in fact it
tends to infinity since $\ew_{\priority(p)}(s_i) = 1$ for infinitely
many $i$. For small $k$ the numerator and the denominator can be equal to
$0$ and then, to avoid all misunderstanding,
it is convenient to assume that the indefinite value
$0/0$ is equal to $-\infty$.)

In the sequel 
the couple $(\poids,\pi)$ consisting of a weight mapping and a
priority mapping will be called  a \emph{weighted priority system}.

Let us note that priority mean-payoff games are a vast generalization
of parity games. In fact parity games correspond to a very particular
case of priority mean-payoff games, we recover the usual  parity games 
when we set for each state $s$,
$\poids(s)=1$ and $\rew(s)=1$ if $\prtmap(s)$ is even and    
$\rew(s)=0$ if $\prtmap(s)$ is odd.

\begin{theorem}\label{theo:optmean}
Priority mean-payoff games over finite arenas admit optimal
deterministic memoryless strategies 
for both players.
\end{theorem}

\begin{proof}
The proof of Theorem~\ref{theo:optmean}
relies on the transfer theorem proved
in~\cite{gimbert:zielonka:2010}. 
This theorem states the following:
if a payoff function $\utlt$ admits
optimal deterministic memoryless strategies
in all one-player perfect information
stochastic games over finite arenas equipped with
payoff $\utlt$ or $-\utlt$, 
then all two-player perfect information stochastic games
over finite arenas with payoff $\utlt$
have also optimal deterministic memoryless strategies for both players.
 
In~\cite{gimbert:zielonka:2007b},
we proved that one-player games 
equipped with the payoff function $\putlt{\rew,\poids,\priority}$
have optimal deterministic memoryless strategy.
 It remains to prove the same for one-player games
equipped with the payoff function $-\putlt{\rew,\poids,\priority}$:
 \begin{equation}
-\putlt{\rew,\poids,\priority}(s_0s_1s_2\cdots)
= \liminf_k 
\frac{
- \sum_{i=0}^k \ew_{\priority(p)}(s) \cdot \poids(s_i) \cdot \rew(s_i)}%
{\sum_{i=0}^k  \ew_{\priority(p)}(s_i)\cdot \poids(s_i)} \enspace .
\label{eq:mpaybis}
\end{equation}
Let us denote $-\rew$ the reward mapping defined by $(-\rew)(s) = -\rew(s)$.
Then, 
 \begin{equation}
\putlt{-\rew,\poids,\priority}(s_0s_1s_2\cdots)
= \limsup_k 
\frac{
- \sum_{i=0}^k \ew_{\priority(p)}(s) \cdot \poids(s_i) \cdot \rew(s_i)}%
{\sum_{i=0}^k  \ew_{\priority(p)}(s_i)\cdot \poids(s_i)} \enspace .
\label{eq:mpayter}
\end{equation}
 The expected values of
 $-\putlt{\rew,\poids,\priority}$
 and $\putlt{-\rew,\poids,\priority}$
 coincide on Markov chains, because in a Markov chain, the limsup
 in~\eqref{eq:mpayter} 
 is almost-surely a limit,
 see the proof of Theorem~7, page~8 of~\cite{gimbert:zielonka:2007b}.
 Since for every play, $-\putlt{\rew,\poids,\priority}(p)\leq
 \putlt{-\rew,\poids,\priority}(p)$, 
 this implies that in a one-player arena, every deterministic memoryless
 strategy optimal for 
 the payoff function $\putlt{-\rew,\poids,\priority}$ is optimal for the
 payoff function $-\putlt{\rew,\poids,\priority}$ as well,
 and these two games have the same values and the same deterministic
 memoryless optimal 
 strategies. 
 This completes the proof.
\end{proof}

\section{From rationally parametrized discounted games to priority
  mean-payoff games}
\label{sec:fromblackwell}

\subsection{Priority mean-payoff derived from rational discount
  parametrization}
The aim of this short subsection is to show how a rational discount
parametrization induces in a canonical way a weighted priority
system. 

Let $\lambda_t$ be a rational discount parametrization.
The fact that $\lim_{t\uparrow 1}(1-\lambda_t(s))=0$  implies that
for each state $s$, the function $t\mapsto 1-\lambda_t(s)$
 factorizes as 
$g_s(t)(1-t)^{\prtmap(s)}$ where $\prtmap(s)\in\NN$ is a positive
integer constant and
$t \mapsto g_s(t)$ is a rational function such that
$g_s(1)\neq 0$.
Moreover since $1-\lambda_t(s)$ is positive for
$t\in(1-\varepsilon,1)$,
$g_s(t)$ is also positive in the same interval and by continuity 
of $g_s(t)$,  $g_s(1) > 0$.

Now, for each state $s$, take $\prtmap(s)$ defined above
as the priority of $s$ and 
$\poids(s):=g_s(1)$ as the weight of $s$.
We say that $(\poids,\priority)$ defined in this way is the
\emph{weighted priority system} derived from the rational
 discount parametrization 
$\lambda_t$.

\subsection{Limit of a discounted game}
The following theorem establishes a remarkable link between discounted
games and weighted priority mean-payoff games. Roughly speaking it
shows that the latter  are the limit of discounted games, the limit not
only in the sense of game values (part (a)) but also 
the optimality of strategies is preserved in the limit. 
\begin{theorem}
\label{th:main}
Let \arn\ be a fixed arena and let
$t\mapsto \lambda_t$ be a rational discount parametrization for \arn.
Let
$(\poids,\pi)$ be the  weighted priority system derived from $\lambda_t$.
Finally let \straopt\ and \strbopt\ be deterministic memoryless Blackwell optimal
strategies for the  discounted game $(\arn,\putlt{\lambda_t})$.

Then
\begin{enumerate}[(a)]
\item
for each state $s$, $\lim_{t\uparrow 1}\val_s(\putlt{\lambda_t}) =
\val_s(\putlt{\rew,\poids,\prtmap})$, where $\val_s(\putlt{\lambda_t})$ is the
value of the game $(\arn,\putlt{\lambda_t})$ and 
$\val_s(\putlt{\rew,\poids,\prtmap})$ is the value of the game
$(\arn,\putlt{\rew,\poids,\prtmap})$,
and
\item
if \straopt\ and \strbopt\ are Blackwell optimal memoryless
deterministic strategies for  the discounted game $(\arn,\putlt{\lambda_t})$ then 
\straopt\ and \strbopt\ are optimal for the priority mean-payoff 
game
$(\arn,\putlt{\rew,\poids,\prtmap})$
\end{enumerate}
\end{theorem} 
Let us note that part (a) of Theorem~\ref{th:main} 
was proved in ~\cite{gimbert:zielonka:2007b} but only for one-player 
games\footnote{In fact, \cite{gimbert:zielonka:2007b}
shows that
the convergence of game values holds not only for rational
parametrizations but for any ``reasonable'' parametrization of
discount factors.}
(Markov decision processes).

However, in~\cite{gimbert:zielonka:2007b}  we were unable to establish any result
linking optimal strategies for discounted games with optimal
strategies of weighted priority games.
Thus the main achievement of the present paper is part (b) of
Theorem~\ref{th:main}.

The following result was proved in~\cite{gimbert:zielonka:2007b}
(Theorem~7 in \cite{gimbert:zielonka:2007b}):
\begin{lemma}
\label{lem:almost}
Let $\lambda_t$ be a rational discount parametrization and let
$(\poids,\prtmap)$ be the derived weighted priority system.
Then for each state $s$ and for all deterministic memoryless
strategies $\sigma,\tau$:
\[ \lim_{t\uparrow 1} \expec{s}{\sigma}{\tau}{\utlt_{\lambda(t)}}
= \expec{s}{\sigma}{\tau}{\putlt{\rew,\poids,\prtmap}}. \]
\end{lemma}


\begin{proof}[Proof of Theorem~\ref{th:main}]

We begin with part (b).
Let \straopt, \strbopt\ be Blackwell optimal deterministic memoryless
strategies
for $\lambda_t$. 
Let $\sigma$ and $\tau$ be any deterministic memoryless
 strategies of players \M\ and
\m.
Then 
\begin{equation*}
\expec{s}{\sigma}{\strbopt}{\putlt{\lambda_t}} \leq
\expec{s}{\straopt}{\strbopt}{\putlt{\lambda_t}} \leq
\expec{s}{\straopt}{\tau}{\putlt{\lambda_t}} .
\end{equation*}
Taking the limit with $t\uparrow 1$ we get by Lemma~\ref{lem:almost}
\[
\expec{s}{\sigma}{\strbopt}{\putlt{\rew,\poids,\prtmap}} \leq
\expec{s}{\straopt}{\strbopt}{\putlt{\rew,\poids,\prtmap}} \leq
\expec{s}{\straopt}{\tau}{\putlt{\rew,\poids,\prtmap}} ,
\]
which shows that $\straopt$ and $\strbopt$ are optimal in the class of
deterministic memoryless strategies.
But Theorem~\ref{theo:optmean} implies that 
for priority mean-payoff games strategies optimal in the
class of deterministic memoryless strategies are optimal also when all
strategies are allowed.
This terminates the proof of (b).

Obviously (a) follows from (b) and from Lemma~\ref{lem:almost}.

\end{proof}

\section{Optimal but not Blackwell optimal strategies}
\begin{figure}[th]
\begin{center}
\includegraphics[scale=0.7]{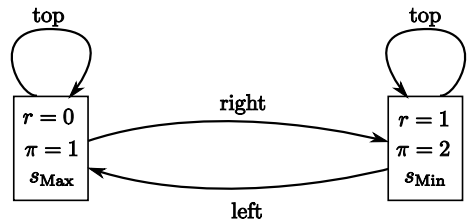}
\end{center}
\caption{A parity game. Player \M\ has two deterministic memoryless
  optimal strategies but only one of them is Blackwell optimal.}
\label{fig:exo}
\end{figure}

Theorem~\ref{th:main} stated that Blackwell optimal strategies are
also optimal for priority mean-payoff games. The converse is not true,
the notion of
Blackwell optimal strategies is strictly more restrictive.

We illustrate this with the game presented in Figure~\ref{fig:exo}.
Here we have two states $s_\M, s_\m$ controlled respectively by
players \M\ and \m. Both states have the same weight $1$ which is
omitted. The left state has priority $\pi=2$ and reward $r=0$, the
right state has priority $\pi=1$ and reward $r=1$, thus essentially
this is the usual parity game with two priorities.
Both players have two deterministic memoryless strategies.
The optimal strategy for player \m\ is to take  action ``left''. 
With this strategy state $s_\M$ with priority $1$ is visited
infinitely often and since this is the minimal priority in this games
the resulting payoff will $0$ whatever the strategy of
player \M.
Player \M\ can play ``top'' or ``right'', in both cases if player
\m\  uses the strategy described above the payoff is $0$ thus both
strategies are optimal for \M.

Now let us consider the associated discounted game with the canonical
parametrization.
Thus the discount factor of $s_\M$ is
$\lambda_t(s_\M)=1-(1-t)^{\pi(s_\M)}=t$
while the discount factor for $s_\m$ is
 $\lambda_t(s_\m)=1-(1-t)^{\pi(s_\m)}=1-(1-t)^2$.
For player \m\ the optimal strategy is still to always play  ``left''.
For player \M\ the strategies ``right'' and ``top'' are now different.
For example if we start from $s_\M$ then playing ``top'' will result
in payoff $0$ since we will visit only the state $s_\M$ with reward
$0$. On the other hand playing ``right'' we will visit infinitely
often the state $s_\m$ with a positive reward, thus for discounted
games playing ``right'' is strictly better for \M\ than playing ``top'' and the
strategy where \M\ plays ``right'' is the only Blackwell optimal
strategy.

The main motivation behind Blackwell optimal strategies comes from the
following observation (due to Blackwell).  Consider a mean-payoff game
controlled completely by player \M\ and suppose that there are only
two possible infinite plays. 
The first play begins with a long but finite
sequence of rewards $0$ followed by an infinite sequence of rewards $1$. The
mean payoff for such history is $1$, the initial sequence of $0$ does
not count on the limit.
Consider now the second play which is an infinite sequence of rewards
$1$, without any $0$. Here also the mean payoff is also $1$. 
Thus player \M\ is indifferent between two histories.
But from the point of view of  Maximizer clearly the second history
is better than the first one, one prefers to have the reward $1$ each day
rather than to begin with the reward $0$.
This difference is captured by
Blackwell optimality.


\bibliographystyle{eptcs}
\bibliography{games}

\end{document}